\documentclass[11pt]{article}
\usepackage[T1]{fontenc}
\usepackage{pslatex}
\usepackage{mathrsfs}
\usepackage{latexsym}
\usepackage{amsmath,amsthm}
\usepackage{amsfonts}
\usepackage{url}
\usepackage{units}
\usepackage{amssymb}
\usepackage{mytheo}
\usepackage{bbm}
\usepackage{graphics}
\usepackage{dsfont}
\newcommand{\chr}{\mathds{1}}

\newcommand{\E}{\mathbb{E}}
\def\clap#1{\hbox to 0pt{\hss#1\hss}}

\newcommand{\inv}{^{-1}}

\newcommand{\kl}{\operatorname{KL}}

\renewcommand{\d}{\mathrm{d}}
\newcommand{\bal}{\beta}

\makeatletter
\newcommand{\ben}{\begin{enumerate}}
\newcommand{\een}{\end{enumerate}}
\newcommand{\bit}{\begin{itemize}}
\newcommand{\eit}{\end{itemize}}

\newcommand{\nrm}[1]{\left\Vert #1 \right\Vert}

\newcommand{\calF}{\mathcal{F}}

\newcommand{\calP}{\mathcal{P}}

\newcommand{\R}{\mathbb{R}}

\newcommand{\beq}{\begin{eqnarray*}}
\newcommand{\eeq}{\end{eqnarray*}}
\newcommand{\beqn}{\begin{eqnarray}}
\newcommand{\eeqn}{\end{eqnarray}}
\newcommand{\paren}[1]{\left( #1 \right)}
\newcommand{\sqprn}[1]{\left[ #1 \right]}
\newcommand{\tlprn}[1]{\left\{ #1 \right\}}
\newcommand{\set}[1]{\tlprn{#1}}
\newcommand{\abs}[1]{\left| #1 \right|}

\newcommand{\gn}{\, | \,}

\newcommand{\hide}[1]{}

\def\eps{\varepsilon}

\newcommand{\sigalg}{\calF}

\renewcommand{\bepf}{\begin{proof}}
\renewcommand{\enpf}{\end{proof}}

\begin{document}
\title{ 
Minimum KL-divergence on complements of $L_1$ balls$^*$
\thanks{
$^*$A previous version had the title ``A Reverse Pinsker Inequality''.
}
}
\author{ Daniel Berend, Peter Harremo{\"e}s, and Aryeh Kontorovich }
\maketitle
\begin{abstract}
Pinsker's widely used inequality upper-bounds the total variation distance $%
\left\Vert P-Q \right\Vert_1$ in terms of the Kullback-Leibler divergence $%
D(P\Vert Q)$. Although in general a bound in the reverse direction is
impossible, in many applications the quantity of interest is actually $%
D^*(v,Q)$ --- defined, for an arbitrary fixed $Q$, as the infimum of $%
D(P\Vert Q) $ over all distributions $P$ that are at least $v$-far away from $Q$ in total
variation. We show that $D^*(v,Q)\le Cv^2 + O(v^3)$, where $C=C(Q)=%
\nicefrac{1}{2}$ for ``balanced'' distributions, thereby providing a kind of
reverse Pinsker inequality. Some of the structural results obtained in the
course of the proof may be of independent interest. An application to large
deviations is given.
\end{abstract}
\section{Introduction}
\subsection{Pinsker's inequality}
The inequality bearing Pinsker's name states that
for two distributions $P$ and $Q$,
\begin{equation}
D(P\Vert Q) \geq \frac{V^{2}( P,Q) }{2},\label{eq:pinsker}
\end{equation}%
where
\begin{equation*}
D(P\Vert Q) =\int \ln\paren{\frac{\d P}{\d Q}} \,\d P
\end{equation*}%
is the Kullback-Leibler divergence of $P$ from $Q$ and 
$V(P,Q)=\left\Vert P-Q\right\Vert _{1}$ is their total variation distance.
Actually, the name is a bit of a misattribution, since the explicit form
of (\ref{eq:pinsker}) was obtained by 
Csisz{\'a}r \cite{MR0219345} and Kullback \cite{kullback67} in 1967
and is occasionally referred to by their names.
Gradual improvements were obtained by \cite%
{MR0214714,MR1984937,MR0270834,MR0252112,0176.49106,MR0214112,MR1873865,MR0479685,MR0275575}
and others; see \cite{DBLP:conf/colt/ReidW09} for a detailed history and the
\textquotedblleft best possible Pinsker inequality\textquotedblright .
Recent extensions to general $f$-divergences may be found in \cite{MR2808583}
and \cite{DBLP:conf/colt/ReidW09}.
This 
inequality has become a ubiquitous tool in
probability \cite{MR815975,MR1739680,marton96}, 
information theory \cite{MR2099014}, and, more recently, machine
learning \cite{MR2409394}. 
It will be useful to define the function $\mathrm{KL}_{2}:(0,1)^{2}%
\rightarrow \lbrack 0,\infty )$ by 
\begin{equation*}
\mathrm{KL}_{2}(p,q)=p\ln {\frac{p}{q}}+(1-p)\ln {\frac{1-p}{1-q}}
\end{equation*}%
and the so-called Vajda's tight lower bound $L$ \cite{MR0275575}: 
\begin{equation*}
L(v)=\inf_{P,Q:V(P,Q)=v}D(P\Vert Q) .
\end{equation*}%
In \cite{MR1984937} an exact parametric equation of the curve $( v,L(v)) _{0<v<\infty
}$ in $\mathbb{R}^{2}$ was given: 
\begin{eqnarray*}
v(t) &=& t\sqprn{ 1-\paren{ \coth t-\frac{1}{t}}^{2}} , \\
L(v(t)) &=&\ln \frac{t}{\sinh t}+t\coth t-\paren{\frac{t}{\sinh t}}^2.
\end{eqnarray*}
Some upper bounds on the KL-divergence in terms of other $f$%
-divergences are known \cite{MR1879596,MR2065310,MR1872472}, and in \cite[Lemma 3.10]{DBLP:conf/ijcai/Even-DarKM07}
it is shown that, under some conditions, $D(P\Vert Q)\le \nrm{P-Q}_1\ln(1/\min Q)$.
The latter
estimate 
is vacuous for $Q$ with infinite support. 
In
general,
it is impossible to upper-bound $D(P\Vert Q) $ in
terms of $V(P,Q)$, since for every $v\in (0,2]$ there is a pair of
distributions $P,Q$ with $V(P,Q)=v$ and $D(P\Vert Q) =\infty $ \cite{MR1984937}.
However, in many applications, the actual quantity of interest is not $%
D(P\Vert Q)$ for arbitrary $P$ and $Q$, but rather 
\begin{equation}
D^*(v,Q)=\inf_{P:V(P,Q)\geq v}D(P\Vert Q) .\label{eq:Ddef}
\end{equation}%
For example, Sanov's Theorem \cite{MR2239987,MR1739680} (which we will
say more about below) implies that the probability that the empirical
distribution $\hat{Q}_{n}$, based on a sample of size $n$, deviates in $\ell
_{1}$ by more than $v$ from the true distribution $Q$ behaves asymptotically
as $\exp ( -nD^*(v,Q)) $.
Throughout this paper, we consider a (finite or $\sigma$-finite)
measure space $(\Omega ,\sigalg ,\mu )$, and all the distributions in
question will be defined on this space and assumed absolutely continuous
with respect to $\mu $; this set of distributions will be denoted by $%
\calP$. 
%The absolute continuity relation will be denoted by $\abscont$.
We will consistently use upper-case letters for distributions $P\in \mathcal{%
P}$ and corresponding lower-case letters for their densities $p$ with
respect to $\mu $. 
We will use standard asymptotic notation $O(\cdot)$ and $\Omega(\cdot)$.
\subsection{Balanced and unbalanced distributions}
In this paper, we show that for the broad class of ``balanced''
distributions, $D^*(v,Q)= v^2/2 +O(v^4)$, which matches the form of the
bound in (\ref{eq:pinsker}). For distributions not belonging to this class,
we show that 
\begin{eqnarray*}
D^*(v,Q)=\frac{v^2}{8\bal(1-\bal)}-O
(v^3),
\end{eqnarray*}
where $\bal$ is a measure of the ``imbalance'' of $Q$ defined below; this
may also be interpreted as a reverse Pinsker inequality.
The \emph{range} of a distribution is 
\begin{equation*}
\mathcal{R}(Q)=\left\{ Q(A):A\in \sigalg \right\} .
\end{equation*}%
A distribution $Q$ has \emph{full range} if $\mathcal{R}(Q)=[0,1]$.
Non-atomic distributions on $\R$ have full range. 
The \emph{balance coefficient} of a distribution $Q$ is 
\begin{eqnarray*}
\bal = \inf\left\{ x\in\mathcal{R}(Q): x\ge\nicefrac{1}{2} \right\}.
\end{eqnarray*}
A distribution is \emph{balanced} if $\bal=\nicefrac{1}{2}$ and \emph{%
unbalanced} otherwise. In particular, all distributions with full range are
balanced. 
Note that the balance coefficient of a  discrete distribution $Q$ is bounded by\footnote{
Since we will not use this fact in the sequel, we only give a proof sketch.
The case where $q_{\max}\ge\nicefrac{1}{2}$ is trivial,
so assume $q_{\max}<\nicefrac{1}{2}$.
Consider the following greedy algorithm: 
initialize $A$ to be
the empty set and
repeatedly include the heaviest available atom such that $A$'s total mass remains
under $\nicefrac{1}{2}$ (once an atom has been added to $A$, it is no longer ``available''). 
If $\omega$ is the first atom 
whose inclusion will
bring 
$A$'s mass over $\nicefrac{1}{2}$,
either $A\cup\set{\omega}$ or $\Omega\setminus A$ establishes the 
bound in (\ref{eq:balqmax}).
}
\beqn
\label{eq:balqmax}
\bal \le  \frac{1}{2} + \frac{q_{\max}}{2}
,
\eeqn
where $q_{\max}=\max_{\omega\in\Omega}q(\omega)$.
Ordentlich and Weinberger \cite{1424321} considered the following
distribution-dependent refinement of Pinsker's inequality. For a
distribution $Q$ with balance coefficient $\beta $, define $\varphi (Q)$ by 
\begin{equation*}
\varphi (Q)=\frac{1}{2\beta -1}\ln \frac{\beta }{1-\beta }
\end{equation*}%
(for $\beta =\nicefrac{1}{2}$, $\varphi (Q)=2$). It is shown in \cite%
{1424321} that 
\begin{equation}
D(P\Vert Q) \geq \frac{\varphi (Q)}{4}V(P,Q)^{2}
\label{eq:OW}
\end{equation}%
for all $P,Q$, and furthermore, that $\varphi (Q)/4$ is the best $Q$%
-dependent coefficient possible: 
\begin{equation}
\inf_{P}\frac{D(P\Vert Q) }{V(P,Q)^{2}}=\frac{%
\varphi (Q)}{4}.\label{eq:OW-tight}
\end{equation}%
Although the left-hand sides of (\ref{eq:Ddef}) and (\ref{eq:OW-tight})
bear a superficial resemblance, the two quantities are quite different (in
particular, the former is constrained by $V(P,Q)\geq v$). While
distribution-independent versions of (\ref{eq:OW}) exist (viz., 
(\ref{eq:pinsker})),
our main result (Theorem \ref{thm:rev-pinsker}) does
not admit a distribution-independent form. Simply put, the result in \cite%
{1424321} yields a lower bound on $D^*(v,Q)$, while we seek to
upper-bound this quantity --- and actually compute it exactly for unbalanced
distributions. 
\section{Main results}
We can now state our reverse Pinsker inequality:
\begin{thm}
\label{thm:rev-pinsker} Suppose $Q\in \calP$ has balance coefficient $%
\beta $. Then:
\begin{enumerate}
\item[(a)] For $\beta \geq \nicefrac{1}{2}$ and $0<v<1$, 
\begin{equation*}
L(v)\leq D^*(v,Q)\leq 
\kl_{2}(\bal -\nicefrac{v}{2},
\bal
) .
\end{equation*}
\item[(b)] For $\beta >\nicefrac{1}{2}$ and 
$0<v<4(\bal-\nicefrac12)$,
\begin{equation*}
D^*(v,Q)=\mathrm{KL}_{2}( \beta -\nicefrac{v}{2},\beta ) .
\end{equation*}
\end{enumerate}
\end{thm}
As a comparison of orders of magnitude, note that 
\begin{eqnarray*}
\mathrm{KL}_{2}( \beta -\nicefrac{v}{2},\beta ) &=&\frac{v^{2}}{%
8\beta (1-\beta )} 
-\frac{(2\beta -1)v^{3}}{48{\beta }^{2}(1-\beta )^{2}}+O(v^{4}), \\
\mathrm{KL}_{2}( \nicefrac{1}{2}-\nicefrac{v}{2},\nicefrac{1}{2})
&=&\frac{v^{2}}{2}+\frac{v^{4}}{12}+O(v^{6}), \\
L(v) &=&\frac{v^{2}}{2}+\frac{v^{4}}{36}+\Omega(v^{6}),
\end{eqnarray*}%
where the 
first two expansions are straightforward and the
last one is well-known \cite{MR1984937}.
Combining the bound of Ordentlich and Weinberger (\ref{eq:OW}) with Theorem %
\ref{thm:rev-pinsker}, we get 
\begin{eqnarray*}
\frac{1}{4(2\beta -1)}\ln \frac{\beta }{1-\beta }v^{2} &\leq &D^*(v,Q)
\\
&\le&\mathrm{KL}_{2}( \beta -\nicefrac{v}{2},\beta ) \\
&=&\frac{v^{2}}{8\beta (1-\beta )}-O( v^{3}) .
\end{eqnarray*}%
As a consistency check, one may verify that 
\begin{equation*}
\frac{1}{4(2\beta -1)}\ln \frac{\beta }{1-\beta }\leq \frac{1}{8\beta
(1-\beta )}
\end{equation*}%
for $\nicefrac{1}{2}\leq \beta <1$.
\begin{thm}
\label{thm:full-range} 
If $Q\in \calP$ has full range, then 
\begin{equation*}
D^*( v,Q) =L(v),
\qquad
0<v<2
.
\end{equation*}
\end{thm}
\section{Proofs}
We will repeatedly invoke the standard fact that $D(\cdot\Vert\cdot)$ is convex in both arguments 
\cite{MR2239987,erven2012}.
Our first lemma provides a structural result for extremal distributions.
Suppose a distribution $Q\in \calP$ is given, along with an $A\in 
\sigalg $ and a $0<v\leq 2( 1-Q(A)) $. 
Denote by $\calP(Q,A,v)$ the set of all distributions $P\in \calP$ for which $V(
P,Q) =v$ and 
$A=\set{\omega\in\Omega: q(\omega)<p(\omega)}$.
The above ``restriction'' on the range of $v$
derives from the fact that every $P\in\calP(Q,A,v)$ must satisfy
$V(P,Q)\le2(1-Q(A))$.
\begin{lem}
\label{lem:PQA} 
For all $Q\in \calP$, $A\in \sigalg$ with $%
0<Q(A)<1 $, and $v\in (0,2(1-Q(A))]$, let $P^*\in \calP$ be the measure with density 
\begin{equation*}
p^*=(a\chr_{A}+b\chr_{\Omega
\setminus A})q,
\end{equation*}%
where 
\begin{equation*}
a=1+\frac{v}{2Q(A)},\qquad b=1-\frac{v}{2(1-Q(A))}.
\end{equation*}%
Then $P^*$ belongs to $\calP(Q,A,v)$, and $P^*$ is the
unique minimizer of $D(P\Vert Q)$ over $P\in \calP(Q,A,v)$.
\end{lem}
\begin{proof}
Obviously, $P^{\ast}$ belongs to $\calP(Q,A,v)$. 
We claim that
\begin{equation}
\label{eq:PP*Q}
D(P\Vert Q) = D( P\left\Vert P^{\ast
}\right. ) +D(P^*\Vert Q)
\end{equation}
holds 
for
all $P\in \calP(Q,A,v)$,
whence the lemma follows.
Indeed, putting $B=\Omega\setminus A$ and using the fact that
\beq
D(P\Vert P^*) \!\!&=& \!\! D(P\Vert Q)-P(A)\ln a - P(B)\ln b,\\
D(P^*\Vert Q) \!\!&=&\!\! Q(A)a\ln a + Q(B)b\ln b,
\eeq
we see that (\ref{eq:PP*Q}) is equivalent to
the identity
$$ (Q(A)-P(A)+\nicefrac{v}{2})\ln a+(Q(B)-P(B)-\nicefrac{v}{2})\ln b =0,$$
which follows immediately from the elementary fact that
$$ P(A)-Q(A)=Q(B)-P(B)=\nicefrac{v}{2}.$$
\end{proof}
Our next result is that $D^*$ actually has a somewhat simpler form than 
the original definition
%in 
(\ref{eq:Ddef}).
\begin{lem}
\label{lem:=eps}
For all distributions $Q$ and all $v>0$, 
\begin{equation*}
D^*(v,Q)=\inf_{P:V(P,Q)=v}D(P\Vert Q) .
\end{equation*}
\end{lem}
\begin{proof}
For any $\eps>0$, let $P_\eps\in\calP$ be such that $V(P_\eps,Q)\ge v$
and
\beqn
\label{eq:Peps}
D(P_\eps\Vert Q)< D^*(v,Q)+\eps
\eeqn
and define,
for $0\le\delta\le1$,
\beq
P_{\eps,\delta} = \delta P_\eps+(1-\delta )Q.
\eeq
Since $V(P_{\eps,\delta},Q)=\delta V(P_{\eps},Q)$,
we may always choose $\overline\delta=\overline\delta(P_\eps)$ so that 
$V(P_{\eps,\overline\delta},Q)=v$.
By convexity of $D(\cdot\Vert\cdot)$, we have
\beq
D(P_{\eps,\overline\delta}\Vert Q) &\le& \overline\delta D(P_\eps\Vert Q) + (1-\overline\delta)D(Q\Vert Q) \\
&\le& D(P_\eps\Vert Q) \\
&<& D^*(v,Q)+\eps,
\eeq
and hence
\beq
\inf_{P:V(P,Q)\ge v}D(P\Vert Q) 
=
\inf_{P:V(P,Q)=v}D(P\Vert Q) .
\eeq
\end{proof}
\begin{proof}[Proof of 
Theorem
\protect\ref{thm:full-range}]
Below we take the infimum over $\calP$ in two steps: 
first over $\calP(Q,A,v)$ and then over $A\in\calF$ satisfying $ Q(A)\le 1-\nicefrac{v}{2}$.
It follows from Lemmas \ref{lem:PQA} and \ref{lem:=eps} that
\beqn
D^*(v,Q) &=& \inf_{P\in\calP:V(P,Q)=v}D(P\Vert Q) \nonumber\\
&=& \inf_{A} \;\inf_{P\in\calP(Q,A,v)} D(P\Vert Q) \nonumber\\
&=& \inf_{A} \;D(P^*\Vert Q) \nonumber\\
&=& \inf_{A} \;\sqprn{ Q(A)a\ln a+Q(\Omega\setminus A)b\ln b} \nonumber\\
&=& \inf_{A} \;\kl_2(Q(A)+\nicefrac{v}{2},Q(A))
\label{eq:verdu},
\eeqn
where $P^*$, $a$ and $b$ are as defined in Lemma \ref{lem:PQA}.
Using the fact \cite{MR1984937,MR2817015} that
\beq
L(v) = \inf_{0<x<1-\nicefrac{v}{2}}\kl_2(x+\nicefrac{v}{2},x),
\eeq
we have that 
$%
D^*(v,Q) = L(v)
$ %
for $Q$ with full range,
which proves 
the claim.
\end{proof}
For the proof of Theorem \protect\ref{thm:rev-pinsker}, we will need
additional lemmata, 
the first of which
will allow us to restrict our attention to distributions with
binary support. Now it is well known \cite{MR1984937} that for each pair of
distributions $P,Q$, there is a pair of binary distributions $P^{\prime
},Q'$ such that $V( P',Q') =V(P,Q)$
and $D(P'\Vert Q') =D(P\Vert Q) $ 
(this fact is generalized to general $f$%
-divergences in \cite{MR2817015}). However, in our case $Q$ is fixed whereas
only $P$ is allowed to vary, and so this result is not directly applicable.
Still, an analogue of this phenomenon also holds in our case.
We will consistently use $\pi$ to denote a map $\Omega\to\set{1,2}$ and
for $Q\in\calP$,
the notation
$\pi(Q)$ refers to the distribution $(Q(\pi\inv(1)),Q(\pi\inv(2)))$ on $\set{1,2}$.
For measurable $A\subseteq\Omega$,
the map 
$\pi_A:\Omega\to\set{1,2}$
is defined by
$\pi_A\inv(1)=A
=\Omega\setminus\pi_A\inv(2)
$.
\begin{lem}
\label{lem:bin-enuf} Let $Q\in\calP$ be a distribution 
whose support
contains at least two points.
Then 
\begin{itemize}
\item[(i)]
For any 
measurable map $\pi:\Omega\to\set{1,2}$ and any
distribution $P'=( p_{1}',p_{2}') $ on $\set{1,2}$,
there exists a $P\in\calP$ such that
$V(P,Q)=V(P',\pi(Q))$ and
$D(P\Vert Q)=D(P'\Vert\pi(Q))$.
In particular,
\beq
D^*(v,\pi(Q))\ge D^*(v,Q),
\qquad v>0.
\eeq
\item[(ii)]
For all $v>0$,
there is a measurable
$\pi:\Omega \rightarrow \{1,2\}$ such that 
$
D^*(v,Q)=D^*( v,\pi (Q)).
$
\end{itemize}
\end{lem}
\begin{proof}
Let $P'=( p_{1}',p_{2}')$ 
be a distribution on $\{1,2\}$
and define the distribution $P\in \calP$ as the mixture 
$$P=p_{1}'Q(\cdot \left\vert \pi ^{-1}(1)\right. ) 
+  p_{2}'Q( \cdot
\left\vert \pi ^{-1}(2)\right. ) .$$ 
Then 
\beq
V(P,Q) &=& \int_\Omega\abs{p(\omega)-q(\omega)}\,\d\mu(\omega) \\
&=& \int_{\pi\inv(1)}\abs{p_1'\frac{q(\omega)}{Q(\pi\inv(1))}-q(\omega)}\,\d\mu(\omega)
+   \int_{\pi\inv(2)}\abs{p_2'\frac{q(\omega)}{Q(\pi\inv(2))}-q(\omega)}\,\d\mu(\omega) \\
&=& Q(\pi\inv(1))\abs{\frac{p_1'}{Q(\pi\inv(1))}-1}
 +  Q(\pi\inv(2))\abs{\frac{p_2'}{Q(\pi\inv(2))}-1} \\
&=& \abs{p_1'-Q(\pi\inv(1))} + \abs{p_2'-Q(\pi\inv(2))} \\
&=& V(P',\pi(Q))
\eeq
and 
\beq
D(P\Vert Q) &=& \int_{\Omega} p(\omega)\log\frac{p(\omega)}{q(\omega)}\, \d\mu(\omega) \\
&=& \int_{\pi\inv(1)} {p_1'\frac{q(\omega)}{Q(\pi\inv(1))}} \log\frac{p_1'q(\omega)/Q(\pi\inv(1))}{q(\omega)}\, \d\mu(\omega)
 +  \int_{\pi\inv(2)} {p_2'\frac{q(\omega)}{Q(\pi\inv(2))}} \log\frac{p_2'q(\omega)/Q(\pi\inv(2))}{q(\omega)} \,\d\mu(\omega) \\
&=& p_1'\log\frac{p_1'}{Q(\pi\inv(1))} + p_2'\log\frac{p_2'}{Q(\pi\inv(2))} \\
&=& D(P'\Vert\pi(Q)).
\eeq
Hence,
\beq
D^*(v,\pi(Q)) &=& \inf_{P':V(P',\pi(Q))=v} D(P'\Vert \pi(Q)) \\
&=& \inf_{
P=p_{1}'Q(\cdot \gn \pi ^{-1}(1)) 
+  p_{2}'Q( \cdot\gn \pi ^{-1}(2))
:
V(P',\pi(Q))=v
} D(P\Vert Q) \\
&\ge&
\inf_{P:V(P,Q)=v} D(P\Vert Q) \\
&=& D^*(v,Q),
\eeq
where the first and last identities follow from Lemma~\ref{lem:=eps}.
This proves (i).
For any $\eps>0$, the proof of Lemma~\ref{lem:=eps}
furnishes a $P_\eps\in\calP$ such that $V(P_\eps,Q)=v$ and
$D(P_\eps\Vert Q)<D^*(v,Q)+\eps$.
Define $\pi $ by 
\begin{equation*}
\pi (\omega )=%
\begin{cases}
1, & \quad\text{$p_\eps(\omega )<q(\omega )$}, \\ 
2, & \quad\text{else}.%
\end{cases}%
~.
\end{equation*}%
Then
\beq
v=V(P_\eps,Q) &=& 
\int_{p_\eps<q}\abs{p_\eps(\omega)-q(\omega)}\,\d\mu(\omega)
+
\int_{p_\eps\ge q}\abs{p_\eps(\omega)-q(\omega)}\,\d\mu(\omega)\\
&=& V(\pi(P_\eps),\pi(Q))
\eeq
and 
\begin{equation*}
D(\pi(P_\eps)\Vert\pi(Q))
\le
D(P_\eps\Vert Q)
<
D^*(v,Q)+\eps,
\end{equation*}
where the first inequality follows from the
data processing inequality \cite[Theorem 9]{erven2012}. 
Since $\eps>0$ is arbitrary, we have that
$D^*(v,\pi(Q))\le D^*(v,Q)$.
Taking $P'=\pi(P_\eps)$, 
it follows from (i) that
$D(P'\Vert\pi(Q))=D(P_\eps\Vert Q)$, which proves (ii).
\end{proof}
Next, we characterize the extremal $P^*$
satisfying $D(P^*\Vert Q) =D^{\ast}(v,Q)$
in the binary case.
\begin{lem}
\label{lem:opt-q-bin} 
Let $Q=( q_{0},1-q_{0}) $ be a binary
distribution with $q_{0}>\nicefrac{1}{2}$ and 
$v\in ( 0,2q_{0}] $. 
Then the unique $P^*$ satisfying $V( P^*,Q) =v$ and 
$D(P^*\Vert Q) =D^*(v,Q)$ is 
\begin{equation*}
P^*=\paren{ q_{0}-\frac{v}{2},1-q_{0}+\frac{v}{2}} .
\end{equation*}
\end{lem}
\begin{proof}
By Lemma \ref{lem:=eps}, there are at most two possibilities for $P^{\ast}$,
namely,
\[
P^{\ast}=P_{1}=\paren{  q_{0}-\frac{v}{2},1-q_{0}+\frac{v}{2}}
\]
and
\[
P^{\ast}=P_{2}=\paren{  q_{0}+\frac{v}{2},1-q_{0}-\frac{v}{2}}  .
\]
(Actually, if $v>2(1-q_{0})$ then only $P_{1}$ is a valid distribution.) 
A second-order Taylor expansion yields
\[
\mathrm{KL}_{2}(  q_{0}+x,q_{0})  =\frac{1}{2}\frac{x^{2}}{(
q_{0}+\theta)  (  1-q_{0}-\theta)  }%
\]
for some $\theta$ between $0$ and $x.$ Hence
\beq
\mathrm{KL}_{2}\paren{  q_{0}-\frac{v}{2},q_{0}}
  <\frac{1}{2}
\frac{(  \nicefrac{v}{2})  ^{2}}{q_{0}(  1-q_{0})} 
<\mathrm{KL}_{2}\paren{  q_{0}+\frac{v}{2},q_{0}}
\eeq
for all $v\in(  0,2(  1-q_{0})  ]  ,$ which implies that
$P^{\ast}=P_{1}$.
\end{proof}
\begin{proof}[Proof of Theorem \protect\ref{thm:rev-pinsker} (a)]
The first inequality is an immediate consequence of Lemma~\ref{lem:=eps}.
To prove the second one,
let $Q\in\calP$ be a distribution 
with balance coefficient $\bal$, and $0<v<1$.
By definition of $\bal$, for all $\eps>0$ there is a measurable $A_\eps\subseteq\Omega$ 
such that
$\bal\le 
Q(A_\eps)
\le\bal+\eps$.
Then Lemma~\ref{lem:bin-enuf}(i) implies that
\beq
D^*(v,Q) \le D^*(v,\pi_{A_\eps}(Q))
\eeq
and by taking $\eps$ arbitrarily small,
\beq
D^*(v,Q) \le D^*(v,Q'),
\eeq
where $Q'=(\bal,1-\bal)$. Finally, Lemma~\ref{lem:opt-q-bin} implies
that $D^*(v,Q')=\kl_{2}(\bal -\nicefrac{v}{2},\bal)$.
\end{proof}
\begin{lem}
\label{lem:voks} 
For every fixed $0<\delta<\nicefrac{1}{2}$,
the binary divergence 
$
\mathrm{KL}_{2}(x-\delta ,x)
$
is strictly increasing in $x$ on 
$\left[ \nicefrac{1}{2}+\nicefrac{\delta}{2},1\right]$. 
\end{lem}
\begin{proof}
Define the function 
\begin{equation*}
F(x)
=\mathrm{KL}_{2}(x-\delta ,x).
\end{equation*}%
Since KL-divergence is jointly convex in the distributions, 
$F$ is a convex function.
Thus, it is sufficient to
prove that 
$F'(x)$
is positive for $x=\nicefrac{1}{2}+%
\nicefrac{\delta}{2}.$ 
We have 
\begin{equation*}
F'(x)
=
\frac{\delta +(1-x)x\ln (1-\frac{\delta }{x}%
)+(1-x)x\ln \frac{1-x}{1-x+\delta }}{(1-x)x}
\end{equation*}%
and
\beq
F'(\nicefrac{1}{2}+\nicefrac{\delta}{2})
&=& \frac{4\delta}{1-\delta^2} +2\log\frac{1-\delta}{1+\delta} 
\;=:\;
G(\delta).
\eeq
Now $G(0)=0$ and
\beq
G'(\delta) = 8\paren{\frac{\delta}{1-\delta^2}}^2 >0,
\eeq
which proves the lemma.
\end{proof}
 
\begin{proof}[Proof of Theorem \protect\ref{thm:rev-pinsker} (b)]
Consider a $Q\in\calP$ with balance coefficient $\bal>\nicefrac{1}{2}$ and
$0<v<4(\bal-\nicefrac12)$.
Then
Lemma \ref{lem:bin-enuf}
implies that
\beq
D^*(v,Q) &=& \inf_{A\in\sigalg}\; D^*(v,\pi_A(Q)) \\
&=& \inf_{A\in\sigalg: Q(A)>\nicefrac{1}{2}}\; D^*(v,(Q(A),(1-Q(A)))),
\eeq
where the second identity holds because
$D^*(v,(q_0,q_1))=D^*(v,(q_1,q_0))$.
Invoking Lemma~\ref{lem:opt-q-bin}, we have that for $Q(A)>\nicefrac12$,
\beq
D^*(v,\pi_A(Q)) = \kl_2\paren{ Q(A)-\frac{v}{2},Q(A)}
\eeq
and hence
\beq
D^*(v,Q) &=& \inf_{A: Q(A)>\nicefrac12}\; \kl_2\paren{ Q(A)-\frac{v}{2},Q(A)}.
\eeq
Since $\nicefrac12+\nicefrac{v}{4}\le\bal\le Q(A)$,
we may invoke Lemma~\ref{lem:voks} with $x=Q(A)$ and $\delta=\nicefrac{v}{2}$
to conclude that 
$D^*(v,Q) = \kl_2\paren{ \bal-\frac{v}{2},\bal}$.
\end{proof}
\section{Application: convergence of the empirical distribution}
The results in \cite{Berend2013} have bearing on the convergence of the empirical
distribution to the true one in the total variation norm. More precisely, 
the paper considers
a sequence of i.i.d. $\mathbb{N}$-valued random variables
$X_{1},X_{2},\ldots$, distributed according to $Q=(q_{1},q_{2},\ldots )$ and
denotes
\begin{equation*}
J_{n}=V( Q,\hat{Q}_{n}) ,\qquad n\in \mathbb{N},
\end{equation*}%
where $\hat{Q}_{n}$ is the empirical distribution induced by the first $n$
observations. Let us recall Sanov's Theorem \cite{MR2239987,MR1739680}, which
yields 
\begin{equation*}
-\lim_{n\rightarrow \infty }\frac{1}{n}\ln Q( J_{n}-\E%
J_{n}>\eps ) =D^*(\eps ,Q).
\end{equation*}
Since the map $( X_{1},\ldots ,X_{n}) \mapsto J_{n}$ is $%
\nicefrac{2}{n}$-Lipschitz continuous with respect to the Hamming distance,
McDiarmid's inequality \cite{mcdiarmid89} implies 
\begin{equation}
\label{eq:BK}
Q( \abs{ J_{n}-\E J_{n}} >\eps )
\leq 2\exp \paren{ -\frac{n\eps ^{2}}{2}},
\qquad
n\in \mathbb{N},\eps>0.
\end{equation}%
Being a rather general-purpose tool, 
in many cases 
McDiarmid's bound does
not yield optimal estimates. Since for balanced
distributions $D^*(\eps ,Q)\leq \eps ^{2}/2+O(
\eps ^{4}) $, we see that the estimate in (\ref{eq:BK})
actually has the optimal constant $\nicefrac{1}{2}$ in the exponent. (See 
\cite[Theorem 1]{MR1822385} for other instances where the quantity $%
\eps ^{2}/2$ emerges in the exponent.) We also see that the McDiarmid's
bound must be suboptimal for unbalanced distributions.
The exponential decrease of $Q( \left\vert J_{n}-\E J_{n}\right\vert >\eps ) $ implies that $J_{n}-\E J_{n}$
tends to zero almost surely. We should note that $\E J_{n}$ will tend
to zero but the rate of convergence may be arbitrarily slow. In \cite%
{Berend2013} it was shown that 
\begin{equation*}
\E J_{n}\leq n^{\nicefrac{\textrm{-}1}{2}}\sum_{j\in \mathbb{N}%
}q_{j}^{\nicefrac{1}{2}}
\end{equation*}%
and that for $Q$ with finite support of size $k$, 
\begin{equation*}
\E J_{n}\leq \paren{ \frac{k}{n}} ^{\nicefrac{1}{2}}.
\end{equation*}%
In greater generality, it was shown that 
\begin{equation*}
{\frac{1}{4}}(\Lambda _{n}-n^{\nicefrac{\textrm{-}1}{2}})\leq \E J_{n}\leq \Lambda _{n},\qquad n\geq 2,
\end{equation*}%
where 
$$
\Lambda _{n}(Q)={n^{\nicefrac{\textrm{-}1}{2}}\sum_{q_{j}\geq 1/n}q_{j}^{%
\nicefrac{1}{2}}} 
+
2\sum_{q_{j}<1/n}q_{j}
$$
tends to zero for $n$ tending to infinity, although the rate at which $%
\Lambda _{n}(Q)$ decays  may be arbitrarily slow, depending on $Q$.
\section*{Acknowledgements}
We thank L{\'a}szl{\'o} Gy{\"o}rfi and Robert Williamson for helpful
correspondence, and in particular for bringing \cite{MR1822385} to our
attention. We are grateful to Sergio Verd\'u for a careful reading
of the paper and useful suggestions.
The comments of the anonymous referees have greatly contributed to the quality of this paper. In particular the proof of 
Theorem~\ref{thm:full-range} has been considerably simplified due to their comments.
\bibliographystyle{plain}
%\bibliography{../../mybib}
\bibliography{mybib}

\begin{thebibliography}{10}

\bibitem{MR2099014}
{Alexander M. Barg, Leonid A. Bassalygo, Vladimir M. Blinovski{\u\i}, et al.}
\newblock In memory of {M}ark {S}em\"enovich {P}insker.
\newblock {\em Problemy Peredachi Informatsii}, 40(1):3--5, 2004.

\bibitem{MR815975}
Andrew~R. Barron.
\newblock Entropy and the central limit theorem.
\newblock {\em Ann. Probab.}, 14(1):336--342, 1986.

\bibitem{MR1822385}
Jan Beirlant, Luc Devroye, L{\'a}szl{\'o} Gy{\"o}rfi, and Igor Vajda.
\newblock Large deviations of divergence measures on partitions.
\newblock {\em J. Statist. Plann. Inference}, 93(1-2):1--16, 2001.

\bibitem{Berend2013}
Daniel Berend and Aryeh Kontorovich.
\newblock A sharp estimate of the binomial mean absolute deviation with
  applications.
\newblock {\em Statistics \& Probability Letters}, 83(4):1254--1259, 2013.

\bibitem{MR2409394}
Nicol{\`o} Cesa-Bianchi and G{\'a}bor Lugosi.
\newblock {\em Prediction, learning, and games}.
\newblock Cambridge University Press, Cambridge, 2006.

\bibitem{MR2239987}
Thomas~M. Cover and Joy~A. Thomas.
\newblock {\em Elements of information theory}.
\newblock Wiley-Interscience, Hoboken, NJ, second edition, 2006.

\bibitem{MR0214714}
Imre Csisz{\'a}r.
\newblock A note on {J}ensen's inequality.
\newblock {\em Studia Sci. Math. Hungar.}, 1:185--188, 1966.

\bibitem{MR0219345}
Imre Csisz{\'a}r.
\newblock Information-type measures of difference of probability distributions
  and indirect observations.
\newblock {\em Studia Sci. Math. Hungar.}, 2:299--318, 1967.

\bibitem{MR1739680}
Frank den Hollander.
\newblock {\em Large deviations}, volume~14 of {\em Fields Institute
  Monographs}.
\newblock American Mathematical Society, Providence, RI, 2000.

\bibitem{MR1879596}
Sever~S. Dragomir.
\newblock Upper bounds for the {K}ullback-{L}eibler distance and applications.
\newblock {\em Bull. Math. Soc. Sci. Math. Roumanie (N.S.)}, 43(91)(1):25--37,
  2000.

\bibitem{MR2065310}
Sever~S. Dragomir and Vido Glu{\v{s}}{\v{c}}evi{\'c}.
\newblock New estimates of the {K}ullback-{L}eibler distance and applications.
\newblock In {\em Inequality theory and applications. {V}ol. {I}}, pages
  123--137. Nova Sci. Publ., Huntington, NY, 2001.

\bibitem{MR1872472}
Sever~S. Dragomir and Vido Glu{\v{s}}{\v{c}}evi{\'c}.
\newblock Some inequalities for the {K}ullback-{L}eibler and
  {$\chi^2$}-distances in information theory and applications.
\newblock {\em Tamsui Oxf. J. Math. Sci.}, 17(2):97--111, 2001.

\bibitem{DBLP:conf/ijcai/Even-DarKM07}
Eyal Even-Dar, Sham~M. Kakade, and Yishay Mansour.
\newblock The value of observation for monitoring dynamic systems.
\newblock In {\em IJCAI}, pages 2474--2479, 2007.

\bibitem{MR1984937}
Alexei~A. Fedotov, Peter Harremo{\"e}s, and Flemming Tops{\o}e.
\newblock Refinements of {P}insker's inequality.
\newblock {\em IEEE Trans. Inform. Theory}, 49(6):1491--1498, 2003.

\bibitem{MR2808583}
Gustavo~L. Gilardoni.
\newblock On {P}insker's and {V}ajda's type inequalities for {C}sisz\'ar's
  {$f$}-divergences.
\newblock {\em IEEE Trans. Inform. Theory}, 56(11):5377--5386, 2010.

\bibitem{MR2817015}
Peter Harremo{\"e}s and Igor Vajda.
\newblock On pairs of {$f$}-divergences and their joint range.
\newblock {\em IEEE Trans. Inform. Theory}, 57(6):3230--3235, 2011.

\bibitem{MR0270834}
Johannes H.~B. Kemperman.
\newblock On the optimum rate of transmitting information.
\newblock In {\em Probability and {I}nformation {T}heory ({P}roc. {I}nternat.
  {S}ympos., {M}c{M}aster {U}niv., {H}amilton, {O}nt., 1968)}, pages 126--169.
  Springer, Berlin, 1969.

\bibitem{MR0252112}
Johannes H.~B. Kemperman.
\newblock On the optimum rate of transmitting information.
\newblock {\em Ann. Math. Statist.}, 40:2156--2177, 1969.

\bibitem{0176.49106}
Olaf Krafft.
\newblock {A note on exponential bounds for binomial probabilities.}
\newblock {\em Ann. Inst. Stat. Math.}, 21:219--220, 1969.

\bibitem{kullback67}
Solomon Kullback.
\newblock A lower bound for discrimination information in terms of variation.
\newblock {\em IEEE Trans. Inform. Theory}, 13:126--127, 1967. Correction,
  volume 16, p. 652, 1970.

\bibitem{marton96}
Katalin Marton.
\newblock {Bounding $\bar{d}$-distance by informational divergence: a method to
  prove measure concentration}.
\newblock {\em Ann. Probab.}, 24(2):857--866, 1996.

\bibitem{mcdiarmid89}
Colin McDiarmid.
\newblock {On the method of bounded differences}.
\newblock In J.~Siemons, editor, {\em Surveys in Combinatorics, volume 141 of
  LMS Lecture Notes Series}, pages 148--188. Morgan Kaufmann Publishers, San
  Mateo, CA, 1989.

\bibitem{MR0214112}
Henry~P. McKean, Jr.
\newblock Speed of approach to equilibrium for {K}ac's caricature of a
  {M}axwellian gas.
\newblock {\em Arch. Rational Mech. Anal.}, 21:343--367, 1966.

\bibitem{1424321}
Erik Ordentlich and Marcelo~J. Weinberger.
\newblock A distribution dependent refinement of pinsker's inequality.
\newblock {\em IEEE Transactions on Information Theory}, 51(5):1836--1840,
  2005.

\bibitem{DBLP:conf/colt/ReidW09}
Mark~D. Reid and Robert~C. Williamson.
\newblock Generalised {P}insker inequalities.
\newblock In {\em COLT}, 2009.

\bibitem{MR1873865}
Flemming Tops{\o}e.
\newblock Bounds for entropy and divergence for distributions over a
  two-element set.
\newblock {\em JIPAM. J. Inequal. Pure Appl. Math.}, 2(2):Article 25, 13 pp.
  (electronic), 2001.

\bibitem{MR0479685}
Godfried~T. Toussaint.
\newblock Probability of error, expected divergence, and the affinity of
  several distributions.
\newblock {\em IEEE Trans. Systems Man Cybernet.}, SMC-8(6):482--485, 1978.

\bibitem{MR0275575}
Igor Vajda.
\newblock Note on discrimination information and variation.
\newblock {\em IEEE Trans. Information Theory}, IT-16:771--773, 1970.

\bibitem{erven2012}
Tim van Erven and Peter Harremo{\"e}s.
\newblock R{\'e}nyi divergence and {K}ullback-{L}eibler divergence.
\newblock {\em IEEE Trans. Inform. Theory}.
\newblock Accepted for publication.

\end{thebibliography}
\end{document}